\documentclass[10pt, conference, compsocconf]{IEEEtran}

\usepackage{cite}

\ifCLASSINFOpdf
\else
\fi

\usepackage[cmex10]{amsmath}
\usepackage{amssymb, amsbsy, amsfonts}
\usepackage{amsthm}
\usepackage{latexsym,float,color}

\theoremstyle{plain}
\newtheorem{theorem}{Theorem}
\newtheorem{lemma}{Lemma}

\newtheorem{proposition}{Proposition}
\theoremstyle{definition}
\newtheorem{definition}{Definition}
\newtheorem{example}{Example}

\newcommand{\SigmaP}{\texttt{Sigma}}
\newcommand{\fct}[3]{#1\colon #2 \to #3}
\newcommand{\dfield}[2]{({#1},{#2})}
\newcommand{\const}[2]{{\rm const}_{#2}{#1}}
\newcommand{\piE}{$\Pi$}
\newcommand{\pisiSE}{$\Pi\Sigma^*$}
\newcommand{\sigmaSE}{$\Sigma^*$}
\newcommand{\rE}{$R^b$}

\newcommand{\vect}[1]{{\bf#1}}
\newcommand{\coeff}{\text{coeff}}

\renewcommand{\AA}{\mathbb{A}}
\newcommand{\NN}{\mathbb{N}}
\newcommand{\ZZ}{\mathbb{Z}}
\newcommand{\QQ}{\mathbb{Q}}
\newcommand{\KK}{\mathbb{K}}
\newcommand{\FF}{\mathbb{F}}
\newcommand{\EE}{\mathbb{E}}
\newcommand{\GG}{\mathbb{G}}
\newcommand{\HH}{\mathbb{H}}

\newcommand{\lr}[1]{\langle #1\rangle}

\makeatletter
\def\@problemhead#1#2#3{%
 \par\kern-\parskip\kern#1
 \vbox\bgroup
 \hbox to\hsize{\hrulefill\raise1pt\hbox{\fbox{$\mathstrut$ #3\ }}\hrulefill}%
 \kern#2\par\kern-\parskip
}
\def\@problemtail#1#2{%
 \par\kern-\parskip\kern#1
 \hbox to\hsize{\hrulefill}
 \egroup
 \kern#2\par\kern-\parskip
}
\newenvironment{ProblemSpec}[1]{\@problemhead{8pt}{2pt}{#1}\small}{\@problemtail{-2pt}{4pt}}
\makeatother

\begin{document}
%
\title{A streamlined difference ring theory: Indefinite nested sums, the alternating sign and the parameterized telescoping problem}


\author{\IEEEauthorblockN{Carsten Schneider}
\IEEEauthorblockA{Research Institute for Symbolic Computation (RISC)\\
Johannes Kepler University (JKU)\\
Linz, Austria\\
Carsten.Schneider@risc.jku.at}}


%


\maketitle

\begin{abstract}
We present an algebraic framework to represent indefinite nested sums over hypergeometric expressions in difference rings. In order to accomplish this task, parts of Karr's difference field theory have been extended to a ring theory in which also the alternating sign can be expressed.
The underlying machinery relies on algorithms that compute all solutions of a given parameterized telescoping equation. As a consequence, we can solve the telescoping and creative telescoping problem in such difference rings.
\end{abstract}


\IEEEpeerreviewmaketitle

\section{Introduction}

The general paradigm of indefinite summation can be specified as follows. Given an expression $F(k)$, find an expression $G(k)$ such that the telescoping equation
\begin{equation}\label{Equ:TeleExpr}
G(k+1)-G(k)=F(k)
\end{equation}
holds. Then with some mild extra conditions on the summation range we can conclude that

\vspace*{-0.3cm}

\begin{equation}\label{Equ:ClosedForm}
G(b+1)-G(a)=\sum_{k=a}^bF(k).
\end{equation}
\noindent More generally, there is the parameterized telescoping problem: given expressions $F_1(k),\dots,F_d(k)$, find an expression $G(k)$ and constants $c_1,\dots,c_d$, free of $k$ and not all zero, such that the parameterized telescoping equation
\begin{equation}\label{Equ:PTeleExpr}
G(k+1)-G(k)=c_1\,F_1(k)+\dots+c_d\,F_d(k)
\end{equation}
holds. Then again by the the telescoping trick we obtain
\begin{equation}\label{Equ:PClosedForm}
G(b+1)-G(a)=c_1\,\sum_{k=a}^bF_1(k)+\dots+c_d\,\sum_{k=a}^bF_d(k).
\end{equation}
As discovered in~\cite{Zeilberger:91} (exploiting Gosper's algorithm for hypergeometric expressions) this paradigm can be utilized to obtain a recurrence for a given sum $S(n)=\sum_{k=a}^bF(n,k)$ depending on an extra discrete parameter $n$. Namely, with the ``creative'' Ansatz $F_i(k)=F(n+i-1,k)$ a solution $c_1,\dots,c_d,G(k)$ for problem~\eqref{Equ:PTeleExpr}
yields the recurrence 

\vspace*{-0.5cm}

\begin{equation}\label{Equ:Rec}
G(b+1)-G(a)=c_1\sum_{k=a}^b\!F(n,k)+\dots+c_d\sum_{k=a}^b\!F(n+d-1,k).
\end{equation}

Note that the expression $G(k)=\sum_{i=a}^{k-1} F(i)$ is trivially a solution of~\eqref{Equ:TeleExpr} which does not deliver any simplification in~\eqref{Equ:ClosedForm}, i.e., both sides of~\eqref{Equ:ClosedForm} are equal. In the same way, one obtains a trivial solution of~\eqref{Equ:PTeleExpr} resp.\ of~\eqref{Equ:PClosedForm}.
In order to hunt for an interesting solution, the tactic ``summation in finite terms'' is often utilized. Here one restricts to a certain ring/field $\AA$ in which $F(k)$ can be represented and develops an algorithm that decides constructively if there exists a solution $G(k)$ of~\eqref{Equ:TeleExpr} that can be represented in $\AA$. 
In this regard, Karr's summation algorithm~\cite{Karr:81,Karr:85} in the setting of difference fields is extremely flexible. Here a
$\Pi\Sigma$-field $\dfield{\AA}{\sigma}$ is introduced, i.e., a field $\AA$ 
is equipped with a field automorphism $\fct{\sigma}{\AA}{\AA}$. There  
the expressions in terms of indefinite nested sums and products are represented in $\AA$, and the shift behaviour of the objects is modelled by $\sigma$. More precisely, if $f\in\AA$ represents the expression $F(k)$, then $\sigma(f)$ represents $F(k+1)$. In this algebraic setting the telescoping problem~\eqref{Equ:TeleExpr} is rephrased as follows: given $\dfield{\AA}{\sigma}$ with $f\in\AA$, find, if possible, a $g\in\AA$ such that $\sigma(g)-g=f$
holds. Similarly, the parameterized telescoping telescoping problem can formulated as follows. Given $f_1,\dots,f_d\in\AA$, find $g\in\AA$ and constants\footnote{I.e.,$\sigma(c_i)=c_i$ for $1\leq i\leq d$.} $c_1,\dots,c_d$, not all $0$, with
\begin{equation}\label{Equ:PT}
 \sigma(g)-\,g=c_1\,f_1+\dots+c_d\,f_d.
\end{equation}


In Karr's algorithm and all the improved variations (see~\cite{Schneider:13a,Schneider:14} and the literature therein) there is one fundamental shortcoming. The alternating sign $(-1)^k$, an important building block in summation formulas, cannot be treated in a difference field: here we are faced with zero divisors such as
$((-1)^k+1)((-1)^k-1)=0$
which can be only treated in rings which are \textit{not} integral domains.
One possibility to overcome this situation is to introduce the concept of interlacing of sequences, resp., of expressions~\cite{Singer:97,Singer:08}. Another (and maybe more natural) approach is to introduce the object $(-1)^k$ as a new summation object and to treat sums and products that are defined over such objects. In~\cite{Schneider:14b} Karr's difference field theory has been generalized to a new difference ring theory that enables one to represent algorithmically indefinite nested sums and products over objects such as $(-1)^{\binom{k+1}{i}}$ with $i\geq0$. However, if one restricts to $(-1)^k$ and slight variations of it, a rather straightforward difference ring theory can be imposed on the already elaborated difference field theory. 

In this article, we will work out these concepts, streamlining the ideas of~\cite{Schneider:05c,Schneider:14b}.
Based on this, we can represent a big class of indefinite nested sums over hypergeometric expressions --also called d'Alembertian expressions~\cite{Abramov:94}, a subclass of Liouvillian expressions~\cite{Singer:99}-- by constructing a tower of difference ring extensions without extending the set of constants.
In particular, we will derive a simplified algorithm for (parameterized) telescoping for such difference rings. Using this toolbox, we will discover as illustrative examples the right hand sides of the following identities: 
\begin{align}
\sum_{k=1}^b & (-1)^k\binom{n}{k}^{-1} 
\sum_{i=0}^{-1+k} \binom{n}{i}=\frac{(-1)^b (b+1)}{(n+2) \binom{n}{b}}\sum_{i=0}^b \binom{n}{i}\nonumber\\[-0.1cm]
&\hspace*{2cm}+\frac{(-1)^b (-2 b-3)}{4 (n+2)}
-\frac{1}{4 (n+2)},\label{Equ:FirstSum}\\
\sum_{k=0}^n& \binom{n}{k} 
\sum_{i=1}^k \frac{(-1)^i}{i}=-2^n 
\sum_{k=1}^n \frac{1}{2^{k}k}.\label{Equ:SecondSum}
\end{align}
The presented algorithms are implemented within the summation package \SigmaP~\cite{Schneider:07a} and are crucial to carry out, e.g., challenging calculations in particle physics; for recent results see~\cite{Physics4} and references therein.

The outline of this article is as follows. In Section~\ref{Sec:PT} we will specify the parameterized telescoping problem. In Section~\ref{Sec:DF} we will present the underlying ideas of the difference field approach. In Sections~\ref{Sec:RExt} and~\ref{Sec:SigmaE} the additional concepts in the setting of difference rings are elaborated and it is worked out how hypergeometric products and indefinite nested sums over such products can be represented in difference rings. Finally, a parameterized telescoping algorithm for the introduced class of difference rings is presented in Section~\ref{Sec:PTAlg}.

\section{The underlying problems PT and FPLDE}\label{Sec:PT}

As motivated in the introduction, we are interested in solving the parameterized telescoping problem in a difference ring (resp. field) $\dfield{\AA}{\sigma}$, i.e., in a ring\footnote{Subsequently, all rings are commutative and all rings (resp.\ fields) contain the rational numbers $\QQ$ as subring (resp.\ subfield). We write $A^*=A\setminus\{0\}$ for a set $A$. $\ZZ$ and $\NN$ denote the set of integers and non-negative integers, respectively.} (resp.\ field) $\AA$ equipped with a ring (resp.\ field) automorphism $\fct{\sigma}{\AA}{\AA}$. Here we define the set of constants by
$\const{\AA}{\sigma}=\{c\in\AA\,|\,\sigma(c)=c\}.$
It is easy to see that $\const{\AA}{\sigma}$ is a subring of $\AA$. 
If we impose that $\const{\AA}{\sigma}$ is a field, we also say that $\const{\AA}{\sigma}$ is the constant field of $\dfield{\AA}{\sigma}$. Note that $\const{\AA}{\sigma}$ is automatically a subfield of $\AA$, if $\AA$ is a field. 

For a difference ring $\dfield{\AA}{\sigma}$ with $\vect{f}=(f_1,\dots,f_d)\in\AA^d$ and $W\subseteq \AA$ we define the solution set
$$V(\vect{f},W)=\{(c_1,\dots,c_d,g)\in(\const{\AA}{\sigma})^d\times W|\, \eqref{Equ:PT}\text{ holds}\}.$$
If $\KK=\const{\AA}{\sigma}$ is a field (and not just a ring), it is easy to verify that $V(\vect{f},\AA)$ is a finite vector space over $\KK$. More generally, if $W$ is a $\KK$-subspace of $\AA$, $V(\vect{f},W)$ is a subspace of $\KK^d\times\AA$ over $\KK$ with dimension $\leq d+1$.  

In summary, to find all solutions of the parameterized telescoping problem can be specified as follows.

\begin{ProblemSpec}{Problem PT in $\dfield{\GG}{\sigma}$ {\footnotesize(Parameterized Telescoping)}}
\noindent \textit{Given} a difference ring (resp.\ field) $\dfield{\GG}{\sigma}$ where $\KK=\const{\GG}{\sigma}$ is a field and given $\vect{f}\in\GG^d$.\\
\noindent \textit{Find} a basis of $V(\vect{f},\GG)$.
\end{ProblemSpec}

In order to solve Problem~PT in the difference field approach~\cite{Karr:81} (i.e., in \pisiSE-fields defined below), in particular, in the difference ring approach (see below), the problem is reduced to a smaller field (resp.\ ring). However, Problem PT cannot always be reduced again to subproblems of type PT. In general it will be reduced to a more general problem that we have to tackle in a difference field $\dfield{\FF}{\sigma}$ with $\KK=\const{\FF}{\sigma}$. Namely, for $a\in\FF^*$ and $\vect{f}=(f_1,\dots,f_d)\in\FF^d$ we define the solution set
$$V(a,\vect{f},\FF)=\{(c_1,\dots,c_d,g)\in(\const{\FF}{\sigma})^d\times\FF|$$ 
\begin{equation}\label{Equ:FPLDE}
 \sigma(g)-a\,g=c_1\,f_1+\dots+c_d\,f_d\}.
\end{equation}
As for Problem PT one can easily check that $V(a,\vect{f},\FF)$ is a subspace of $\KK^d\times\FF$ over $\KK$ with dimension $\leq d+1$. Summarizing, we are interested in the following problem.

\begin{ProblemSpec}{Problem FPLDE in $\dfield{\GG}{\sigma}$ \begin{minipage}[t]{3.2cm}{\footnotesize(First-order Parameterized\\[-0.13cm] Linear Difference Equ.)}\end{minipage}} 
\noindent \textit{Given} a difference field $\dfield{\GG}{\sigma}$ with $a\in\GG^*$ and $\vect{f}\in\GG^d$.\\
\textit{Find} a basis of $V(a,\vect{f},\GG)$.
\end{ProblemSpec}

\section{The difference field approach}\label{Sec:DF}

We aim at solving a parameterized telescoping equation~\eqref{Equ:PTeleExpr}
in terms of a difference field (resp. ring). Here we are faced with three subproblems. 
\begin{enumerate}
\item Construct a difference field (resp.\ ring) in which the summation objects are modelled accordingly. 
\item Solve Problem PT in this setting. 
\item Rephrase the solution in terms of sums and products.
\end{enumerate}
We will illustrate this procedure in the setting of difference fields by discovering the identity
\begin{equation}\label{Equ:SimpleBinom}
\text{\small$\displaystyle\sum_{k=1}^b$} \Big(
        \text{\small$\displaystyle\sum_{i=0}^{k-1}$}\tbinom{n}{i}\Big)\\
=\frac{1}{2} (2 b-n) 
\text{\small$\displaystyle\sum_{i=0}^b$}\tbinom{n}{i}
+\frac{1}{2} \tbinom{n}{b} (-b
+n).
\end{equation}
Concerning Subproblem 1, a difference field $\dfield{\FF}{\sigma}$ is constructed by adjoining step by step new variables that describe the arising summation objects, and the automorphism, acting on the variables, is extended accordingly in order to model the shift behaviour of the summation objects. Here we exploit the following basic lemma.

\begin{lemma}
Let $\dfield{\FF}{\sigma}$ be a difference field with $\alpha\in\FF^*$, $\beta\in\FF$ and let $t$ be transcendental over $\FF$, i.e., $\FF(t)$ is a rational function field. Then there is a unique field automorphism $\fct{\sigma'}{\FF(t)}{\FF(t)}$ with $\sigma'|_{\FF}=\sigma$ and $\sigma'(t)=\alpha\,t+\beta$.
\end{lemma}

\begin{example}\label{Exp:PiSi1}
We will represent the summand $F(k)=\sum_{i=0}^{-1+k} \binom{n}{i}$ given in~\eqref{Equ:SimpleBinom} in a difference field.\\
(0) Take the rational function field $\QQ(n)$ and the automorphism $\fct{\sigma}{\QQ(n)}{\QQ(n)}$ with $\sigma(c)=c$ for all $c\in\QQ(n)$.\\
(1) Take the rational function field $\QQ(n)(k)$ and extend $\sigma$ from $\QQ(n)$ to $\QQ(n)(k)$ with $\sigma(k)=k+1$.\\
(2) Take the rational function field $\QQ(n)(k)(b)$ and represent $\binom{n}{k}$ by $b$. Since $\binom{n}{k+1}=\frac{n-k}{k+1}\binom{n}{k}$, we extend $\sigma$ from $\QQ(n)(k)$ to $\QQ(n)(k)(b)$ with $\sigma(b)=\frac{n-k}{k+1}\,b.$\\
(3) Take the function field $\QQ(n)(k)(b)(s)$ and represent the sum $F(k)$ by $s$. Since $F(k+1)=F(k)+\binom{n}{k}$, we extend $\sigma$ from $\QQ(n)(k)(b)$ to $\QQ(n)(k)(b)(s)$ with $\sigma(s)=s+b.$ 
\end{example}

In a nutshell, we constructed a tower of difference field extensions. Here,
a difference field $\dfield{\EE}{\sigma'}$ is called a difference field extension of $\dfield{\FF}{\sigma}$ (in short, $\dfield{\FF}{\sigma}\leq\dfield{\EE}{\sigma'}$) if $\EE$ is a field extension of $\FF$ and $\sigma'|_{\FF}=\sigma$; subsequently, we do not distinguish anymore between $\sigma$ and $\sigma'$. E.g., in our example we built the following tower of extensions:
\small
\begin{align*}
\dfield{\QQ(n)}{\sigma}&\leq\dfield{\QQ(n)(k)}{\sigma}\\
&\leq\dfield{\QQ(n)(k)(b)}{\sigma}\leq\dfield{\QQ(n)(k)(b)(s)}{\sigma}.
\end{align*}
\normalsize
In addition, we have the property that during the extensions the constants remain unchanged. Namely for the constructed field $\FF=\QQ(n)(k)(b)(s)$ we have that
$\const{\FF}{\sigma}=\QQ(n)$.
Exactly this type of extensions is called \pisiSE-extension~\cite{Karr:81}.

\begin{definition}\label{Def:PiSi}
Consider the difference field extension $\dfield{\FF(t)}{\sigma}$ of $\dfield{\FF}{\sigma}$ with $t$ transcendental over $\FF$, $\sigma(t)=\alpha\,t+\beta$ where $\alpha\in\FF^*$ and $\beta\in\FF$, and $\const{\FF(t)}{\sigma}=\const{\FF}{\sigma}$. 
\begin{enumerate}
\item This extension is called a \piE-extension if $\beta=0$. 
\item It is called a \sigmaSE-extension if $\alpha=1$.
\item It is a \pisiSE-extension if it is a \piE- or \sigmaSE-extension.
\end{enumerate}
Finally, a tower of \piE-/\sigmaSE-/\pisiSE-extensions is called a (nested) \piE-/\sigmaSE-/\pisiSE-extension. A difference field $\dfield{\FF}{\sigma}$ with constant field $\KK$ is called \pisiSE-field over $\KK$ if $\dfield{\FF}{\sigma}$ is a (nested) \pisiSE-extension of $\dfield{\KK}{\sigma}$.
\end{definition}

\begin{example}[Cont.\ Ex.\ \ref{Exp:PiSi1}]
Take our \pisiSE-field $\dfield{\FF}{\sigma}$ over $\QQ(n)$.
The summand $F(k)$ in~\eqref{Equ:SimpleBinom} can be represented by $f=s\in\FF$. With the summation package \SigmaP\ we get the solution 
$g=\frac{b\,k}{2}
+\frac{1}{2} s (-2
+2 k
-n
)\in\FF$
of $\sigma(g)-g=f$. 
This gives the solution 
$G(k)=\frac{1}{2} (-2
+2 k
-n
) 
\sum_{i=0}^{-1+k} \binom{n}{i}
+\frac{k}{2}\binom{n}{k}$
for~\eqref{Equ:TeleExpr}, and by the telescoping trick we arrive at~\eqref{Equ:SimpleBinom}.
\end{example}

\textit{Remark on Subproblem 1.} It is not obvious why the constructed difference field $\dfield{\FF}{\sigma}$ is a \pisiSE-field and how we can design such a difference field automatically in which the summand $F(k)$ can be represented. In~\cite{Karr:85} Karr derived tools that enable one to check if a tower of extensions built by variables is indeed a \pisiSE-field. In addition, in~\cite{Schneider:10c,Schneider:13a} it has been elaborated how these tools can be used to perform such constructions almost automatically. However, only with the possibility to treat objects like $(-1)^k$ algorithmically, this approach turns out be fully algorithmic. Further details on these aspects will be given in  Section~\ref{Sec:RExt}.

\textit{Remark on Subproblem 2.} In general, given a \pisiSE-field $\dfield{\FF}{\sigma}$ over a constant field with certain algorithmic properties, Karr's summation algorithm~\cite{Karr:81} solves Problem~PT in $\dfield{\FF}{\sigma}$; for a simplified and improved version implemented in \SigmaP\ we refer to~\cite{Schneider:14}. More generally, as already indicated in Section~\ref{Sec:DF}, these algorithms rely on a solution of Problem FPLDE. 
In the \pisiSE-field setting, due to~\cite{Karr:81} and~\cite[Thm.~3.5]{Schneider:05c} we obtain the following result.

\begin{theorem}\label{Thm:FPLDEinPiSi}
Let $\dfield{\FF}{\sigma}$ be a \pisiSE-field where the constant field is given by a rational function field $\GG(y_1,\dots,y_o)$ defined over an algebraic number field $\GG$. Then there is an algorithm that solves problem FPLDE in $\dfield{\FF}{\sigma}$.  
\end{theorem}


\textit{Remark on Subproblem 3.} The reformulation of the difference field solution is carried out by reinterpreting the corresponding variables as sums and products. A rigorous translation mechanism is introduced in~\cite{Schneider:10b,Schneider:10c,Schneider:13a}.

\section{Root of unity extensions and the representation of products}\label{Sec:RExt}

In this section we show how objects like $(-1)^k$ can be represented in a difference ring. Using this construction together with \piE-extensions we will show afterwards how a finite set of hypergeometric expressions can be represented in a tower of difference ring extensions without enlarging the underlying constant field.
In this regard, a difference ring $\dfield{\EE}{\sigma'}$ is a difference ring extension of $\dfield{\AA}{\sigma}$ if $\EE$ is a ring extension of $\AA$ and $\sigma'|_{\AA}=\sigma$; as with the field version we do not distinguish anymore between $\sigma'$ and $\sigma$.

In the following, let $\dfield{\AA}{\sigma}$ be a difference ring (or field) with constant field $\KK$ and let $\alpha\in\KK^*$ be a primitive $\lambda$-th root of unity with $\lambda>1$. Note that
$\alpha^{\lambda}=1$ where $\lambda$ is minimal. A typical example is $\alpha=-1$ with $\lambda=2$ or the imaginary part $\alpha=i$ with $\lambda=4$. We will construct a difference ring extension in which the object $(\alpha)^k$ can be represented, i.e.,
where the properties $((\alpha)^k)^{\lambda}=1$ and $(\alpha)^{k+1}=\alpha\,(\alpha)^k$ are rephrased algebraically. 

First, take the difference ring extension
$\dfield{\AA[y]}{\sigma}$ of $\dfield{\AA}{\sigma}$ with $y$ being transcendental over $\AA$ and $\sigma(y)=\alpha\,y$ (again this construction is unique).
Now
take the ideal $I:=\lr{y^{\lambda}-1}$ and consider the quotient ring $\EE=\AA[y]/I$.
Since $I$ is closed under $\sigma$, i.e., $I$ is a reflexive difference ideal, one can
verify that $\fct{\sigma}{\EE}{\EE}$ with
$\sigma(f+I)=\sigma(f)+I$
forms a ring automorphism. In other words, $\dfield{\EE}{\sigma}$ is a
difference ring. Moreover, there is the natural embedding of $\AA$ into $\EE$
with
$a\mapsto a+I$.
By identifying $a$ with $a+I$, $\dfield{\EE}{\sigma}$ is a difference ring extension of $\dfield{\AA}{\sigma}$. Finally, by setting $x:=y+I$. we get the difference ring extension $\dfield{\AA[x]}{\sigma}$ of $\dfield{\AA}{\sigma}$ subject to the relation $x^{\lambda}=1$. This extension is also called algebraic extension.
Note that $x^{\lambda}=1$ and $\sigma(x)=\alpha\,x$ model exactly the object $\alpha^k$ with $((\alpha)^k)^{\lambda}=1$ and $(\alpha)^{k+1}=\alpha\,(\alpha)^k$.
To this end, we are interested in those extensions where the constants remain unchanged; see~\cite{Schneider:14b}.

\begin{definition}
Let $\dfield{\AA}{\sigma}$ be a difference field and let
$\alpha\in\const{\AA}{\sigma}$ be a primitive $\lambda$-th root of unity ($\lambda>1$). Then the algebraic extension $\dfield{\AA[x]}{\sigma}$ of $\dfield{\AA}{\sigma}$ is called basic root of unity extension (\rE-extension) if $\const{\AA[x]}{\sigma}=\const{\AA}{\sigma}$.
\end{definition}

\noindent To check if the constants remain unchanged can be non-trivial~\cite{Schneider:14b}. However, one can always construct an \rE-extension over a difference field which is constant-stable.

\begin{definition}
 A difference field/ring $\dfield{\AA}{\sigma}$ is constant-stable if for any $k>0$ we have that $\const{\AA}{\sigma^k}=\const{\AA}{\sigma}$.
\end{definition}

\noindent Namely, we get the following result.

\begin{proposition}\label{Prop:RExt}
Let $\dfield{\FF}{\sigma}$ be a constant-stable difference field. Let $\alpha\in\const{\FF}{\sigma}$ be a primitive $\lambda$-th root of unity ($\lambda>1$). The algebraic extension from above is an \rE-extension. 
\end{proposition}
\begin{proof}
Suppose there is a $g=\sum_{i=0}^{\lambda-1}g_i\,x^i\in\FF[x]\setminus\FF$ with $\sigma(g)=g$. Thus there is an $m$ with $0<m<\lambda$ with $g_m\neq0$. By comparing the $m$-th coefficient in $\sigma(g)=g$ it follows that $\sigma(g_m)=\alpha^m\,g_m$.
Since $\alpha$ is a constant, $\sigma^{\lambda}(g_m)=(\alpha^m)^{\lambda}\,g_m=g_m$. Hence $g_m\in(\const{\FF}{\sigma})^*$ since  $\dfield{\FF}{\sigma}$ is constant-stable. Consequently, $g_m=\sigma(g_m)=\alpha^m\,g_m$. Since $g_m$ is invertible, we get $\alpha^m=1$; this contradicts to the minimality of $\lambda\neq0$ for $\alpha^{\lambda}=1$.
\end{proof}

\smallskip

\noindent\textit{Summary:} Since any \pisiSE-field is constant-stable, see~\cite{Karr:81}, we can always construct an \rE-extension over a \pisiSE-field.


\subsection{Representation of hypergeometric products}

In the following we restrict to a \pisiSE-field $\dfield{\KK(k)}{\sigma}$ with $\sigma(k)=k+1$ where the constant field $\KK=\QQ(y_1,\dots,y_o)$ is a rational function field ($o\geq0$). Then we want to solve the following problem\footnote{Let $\alpha\in\KK(k)$ and let $\KK(k)$ be a subfield of $\EE$ with $u\in\EE$. In the following $\alpha(u)$ means that we replace in $\alpha$ any occurrence of $k$ by $u$.} (see subproblem 1 in Section~\ref{Sec:DF}).

\begin{ProblemSpec}{Problem RHP {\footnotesize(Represent hypergeometric products)}}
\noindent \textit{Given} the hypergeometric products over $\alpha_i\in\KK(k)^*$:

\vspace*{-0.3cm}

\begin{equation}\label{Equ:ProdSet}
P_1(k)=\prod_{j=\lambda_1}^{k}\alpha_1(j),\dots,P_n(k)=\prod_{j=\lambda_n}^{k}\alpha_n(j)
\end{equation}

\vspace*{-0.2cm}

\noindent with $\lambda_i\in\NN$ where $\alpha_i(j)$ has no pole and is non-zero for $j\geq\lambda_i$.\\
\textit{Find} a difference ring extension $\dfield{\AA}{\sigma}$ of $\dfield{\KK(k)}{\sigma}$ with $\const{\AA}{\sigma}=\KK$ such that for all $1\leq i\leq n$
there is a $g_i\in\AA^*$ with
$\sigma(g_i)=\alpha_i\,g_i.$
\end{ProblemSpec}

\noindent I.e., we want to model a finite set of products~\eqref{Equ:ProdSet} with the shift-behaviour 
$\prod_{j=\lambda_i}^{k+1}\alpha_i(j)=\sigma(\alpha_i)\,\prod_{j=\lambda_i}^{k}\alpha_i(j)$
in a ring extension without extending the constants $\KK$.

A solution to this problem has been presented in~\cite{Schneider:05c}. A streamlined version can be given as follows. 
Here we define the polynomial ring $R=\ZZ[y_1,\dots,y_o,k]$. Note that the gcd in $R$ is uniquely determined up to the units $1,-1$ of $R$.

\smallskip

\textit{Step (1).} Factorize all $\alpha_i$ into irreducible factors, i.e.,
\begin{equation}\label{Equ:FullFactor}
\alpha_i=f_{i,1}^{m_{i,1}}\dots f_{i,n_i}^{m_{i,n_i}}\text{ with } m_{i,j}\in\ZZ
\end{equation}
and irreducible $f_{i,1},\dots,f_{i,n_i}\in R$ being pairwise co-prime.

Among the $\alpha_i$ there might be several factors which are shift-equivalent. More precisely, two irreducible factors $a,b\in R$ are called shift-equivalent if there is an $r\in\ZZ$ with $\gcd(a,\sigma^r(b))\neq 1$.
Otherwise, $a$ and $b$ are called shift-prime.\\
The following observations are immediate: If both elements are free of $k$, they are shift-prime iff they are the same up to the unit $-1$. If one element is free of $k$ and the other element depends on $k$, they are shift-prime. Otherwise, it both depend on $k$, it is well known that $\gcd(\sigma^r(a),b)\neq1$ iff $r\in\ZZ$ is a root of $p(z)=res_k(a(z),b(k+z))\in\KK[z]$. Summarizing, we can decide algorithmically if two irreducible factors $a,b\in R$ are shift-prime. If not, we obtain in addition a witness $r\in\ZZ$ for $\gcd(a,\sigma^r(b))\neq 1$.
In this case, we even have 
$a=(\pm1)\sigma^r(b)$, since $a,b$ are irreducible.

\smallskip

\textit{Step (2).} With this algorithmic property we can proceed as follows.
Rewrite for $1\leq i\leq n$ the factorization~\eqref{Equ:FullFactor} to
\begin{equation}\label{Equ:CompactRep}
\alpha_i=\frac{\sigma(g_i)}{g_i}(-1)^{w_i} h_{1}^{\mu_{i,1}}\dots h_{e}^{\mu_{i,e}}
\end{equation}
with $w_i\in\{0,1\}$, $\mu_{i,j}\in\ZZ$, $g_i\in\KK(k)^*$ and with $h_1,\dots,h_e\in R$ being pairwise shift-prime; note that the $h_1,\dots,h_e$ are used simultaneously for all the $\alpha_i$.\\ 
Namely, pick out a factor $f_{i,j}$, not treated so far, and determine the factors in all the $\alpha_i$ which are not shift-prime. 
Then one can exploit the following lemma to express all these shift-equivalent factors by the representant\footnote{Alternatively, one can take any representant $\sigma^u(f_{i,j})$ with $u\in\ZZ$. This might lead to compacter expressions in~\eqref{Equ:CompactRep} which is closely related to~\cite{Petkov:10}.} $f_{i,j}$ multiplied by a factor $\sigma(\gamma)/\gamma$ for some $\gamma\in\KK(k)^*$.

\begin{lemma}\label{Lemma:MapFactors}
Let $\dfield{\KK(k)}{\sigma}$ be a \pisiSE-field over $\KK$ with $\sigma(k)=k+1$ and $p,q\in\KK[k]^*$ with $\sigma^r(p)=q$
for some $r\in\ZZ$. Then one can construct a $\gamma\in\KK(k)^*$ with
$q=p\,\frac{\sigma(\gamma)}{\gamma}$.
\end{lemma}
\begin{proof}
If $r\geq0$, take $\gamma:=\prod_{i=0}^{r-1} \sigma^{i}(p)\in\KK[k]$.
Then $\frac{\sigma(\gamma)}{\gamma}=\frac{\prod_{i=0}^{r-1}
\sigma^{i+1}(p)}{\prod_{i=0}^{r-1} \sigma^i(p)}
=\frac{\sigma^r(p)}{p}=\frac{q}{p}$. Similarly, if $r<0$, take
$\gamma:=\prod_{i=1}^{-r} \sigma^{-i}(\frac{1}{p})$ and we get that
$\frac{\sigma(\gamma)}{\gamma}=\frac{q}{p}$.
\end{proof}

\begin{example}\label{Exp:Pi1}
Consider, e.g., the  products
$$P_1(k)=\text{\small$\prod_{j=1}^k$}\alpha_1(j)\text{ and }P_2(k)=\text{\small$\prod_{j=1}^k$}\alpha_2(j)$$
with 
$\alpha_1=2(n+1)^2(k+n)$ and $\alpha_2=2^2(n+1)(-k-2-n)k$.
Then among all the factors $2,n+1,k+n,-k-2-n,k$ in $\alpha_1$ and $\alpha_2$, only $k+n$ and $-k-2-n$ are not shift-prime:
$\sigma^2(k+n)=-(-k-n-2)$. By Lemma~\ref{Lemma:MapFactors} we get
$\frac{\sigma(g)}{g}(k+n)=-(-k-n-2)$
with $g=(k+n)(k+n+1)$. Thus we can write
$\alpha_1=2(n+1)^2(k+n)$ and
$\alpha_2=-2^2(n+1)(k+n)k\tfrac{\sigma(g)}{g}$ 
by using $2$, $k$, $k+n$ and the unit $-1$. 
\end{example}

Eventually, all factors in the representations~\eqref{Equ:FullFactor} are treated and what remains in each of the representations is a unit of $R$, i.e., $(-1)^{w_i}$ with $w_i\in\{0,1\}$. Summarizing, we end up at the representation~\eqref{Equ:CompactRep}. In particular, we get the following alternative representation of $P_1(k),\dots,P_n(k)$.
Namely, let $\lambda'_i\geq\lambda_i$ be sufficiently large such that for all $j\geq\lambda'_i$ we have that $g(j)$ has no pole and is not zero.
Then for $1\leq i\leq n$, define $c'_i=\prod_{j=\lambda_i}^{\lambda'_i-1}\alpha_i(j)\in\KK^*$ and $c_i=\frac{c'_i}{g_i(\lambda'_i)}\in\KK^*$. Thus by multiplicative telescoping,

\vspace*{-0.3cm}

\small
\begin{align}\label{Equ:ProdGoodRep}
&P_i(k)=
c'_i\,\tfrac{g_i(k+1)}{g_i(\lambda'_i)}\prod_{j=\lambda'_i}^k\big((-1)^{w_i}\,h_1(j)^{\mu_{i,1}}\dots h_e(j)^{\mu_{i,e}}\big)\\[-0.2cm]
&=c_i\,g_i(k+1)((-1)^k)^{w_i}\big(\prod_{i=\lambda'_i}^kh_1(j)\big)^{\mu_{i,1}}\dots\big(\prod_{i=\lambda'_i}^kh_e(j)\big)^{\mu_{i,e}}.\nonumber
\end{align}
\normalsize

\vspace*{-0.1cm}

\textit{ Step (3).} Finally, we construct a difference ring in which we can express the products in a \piE-extension together with an \rE-extension. Here we rely on the following lemma which follows by iterative application of~\cite[Thm.6.9]{Schneider:05c}.

\begin{lemma}\label{Lemma:GetPiExt}
Let $\dfield{\KK(k)}{\sigma}$ be the \pisiSE-field over a rational function field $\KK=\QQ(y_1,\dots,y_o)$ with $\sigma(k)=k+1$. Let $h_1,\dots,h_e\in\ZZ[y_1,\dots,y_o,k]$ be irreducible and pairwise shift-prime.
Then there is a \piE-extension $\dfield{\KK(k)(t_1)\dots(t_e)}{\sigma}$ of $\dfield{\KK(k)}{\sigma}$ with $\frac{\sigma(t_i)}{t_i}=h_i$.
\end{lemma}

Note that the elements $\sigma^{u_1}(h_1),\dots,\sigma^{u_e}(h_e)$ with $u_i\in\ZZ$ are also irreducible elements from $R$ and are also pairwise shift-prime.
As a consequence,
we can take, e.g., the \piE-extension $\dfield{\KK(k)(t_1)\dots(t_e)}{\sigma}$ of $\dfield{\KK(k)}{\sigma}$ with $\sigma(t_i)=\sigma(h_i)\,t_i$ by Lemma~\ref{Lemma:GetPiExt}. Moreover, since $\dfield{\KK(k)}{\sigma}$ is constant stable, we can invoke Proposition~\ref{Prop:RExt}. Hence
we can construct the \rE-extension $\dfield{\AA}{\sigma}$ of $\dfield{\KK(k)(t_1)\dots(t_e)}{\sigma}$ with $\AA=\KK(k)(t_1)\dots(t_e)[x]$ and $\sigma(x)=-x$. Then by~\eqref{Equ:ProdGoodRep} it is immediate that we can model $P_i(k)$ by
$c_i\,\sigma(g_i)x^{w_i}t_1^{\mu_{i,1}}\dots t_e^{\mu_{i,e}}\in\AA^*$ for all $1\leq i\leq n$.

\begin{example}[Cont. Ex.~\ref{Exp:Pi1}]\label{Exp:Pi2}
We obtain
\begin{align*}
P_1(k)&=
2^k((n+1)^k)^2\Big(\prod_{i=1}^k(i+n)\Big),\\[-0.4cm]
P_2(k)&=
\tfrac{(k+n+2)(k+n+1)}{(n+1)(n+2)}(-1)^k(2^k)^2\Big(\prod_{i=1}^k(i+n)\Big)k!\,.
\end{align*}
By construction the elements $h_1=2$, $h_2=n+1$, $h_3=n+k$, $h_4=k$ in $\ZZ[n,k]$ are shift-prime. Hence we can construct the \piE-extension
$\dfield{\QQ(n)(k)(t_1)(t_2)(t_3)(t_4)}{\sigma}$ of $\dfield{\QQ(n)(k)}{\sigma}$ with $\sigma(t_i)=\sigma(h_i)\,t_i$. Finally, we can construct the \rE-extension $\dfield{\QQ(n)(k)(t_1)(t_2)(t_3)(t_4)[x]}{\sigma}$ of $\dfield{\QQ(n)(k)(t_1)(t_2)(t_3)(t_4)}{\sigma}$ with $\sigma(x)=-x$. Here we can choose
$g_1=t_1\,t_2^2\,t_3\,t_4$ and 
$g_2=\tfrac{(k+n+2)(k+n+1)}{(n+1)(n+2)}\,x\,t_1^2\,t_3\,t_4$
to model $P_1(k)$ and $P_2(k)$ accordingly.
\end{example}

Summarizing, we end up at the following theorem; compare~\cite[Cor.~6.12]{Schneider:05c}.

\begin{theorem}\label{Thm:RepresentHyp}
Let $\dfield{\KK(k)}{\sigma}$ be a \pisiSE-field over the rational function field $\KK=\QQ(y_1,\dots,y_o)$ with $\sigma(k)=k+1$. Let $\alpha_1,\dots,\alpha_n\in\KK(k)^*$. Then there is an algorithm that solves Problem RHP. More precisely, the following holds.\\
(1) It computes a \piE-extension $\dfield{\HH}{\sigma}$ of $\dfield{\KK(k)}{\sigma}$ with $\HH=\KK(k)(t_1)\dots(t_e)$ where the
$\sigma(t_i)/t_i\in\ZZ[y_1,\dots,y_o,k]$ with $1\leq i\leq e$ are irreducible and pairwise shift-prime.\\
(2) It computes an \rE-extension $\dfield{\AA}{\sigma}$ of $\dfield{\HH}{\sigma}$ with $\AA=\HH[x]$ and $\sigma(x)=-x$.\\
(3) For any $1\leq r\leq n$, it computes a $g\in\AA^*$ with $\sigma(g)=\alpha_i\,g$;
here $g$ has the form $q\,x^w\,t_1^{\mu_{1}}\dots t_e^{\mu_{e}}$ with $w\in\{0,1\}$,$\mu_{1},\dots,\mu_e\in\ZZ$ and $q\in\KK(k)^*$.
\end{theorem}

We conclude this section by two remarks.\\ 
(1) The construction can be performed by taking step by step the different products $P_i(k)$. Hence, whenever a new product arises within a summation problem, we can continue this construction and will always succeed in adjoining new \piE-extensions such that also this product can be represented in the possibly enlarged difference ring.\\
(2) This flexibility is guaranteed by taking only multiplicands $\sigma(t_i)/t_i$ that are irreducible polynomials from $R$. But, it is often desirable to merge different \piE-generators to one element. Here the following result is useful; the proof is left to the reader and follows easily by using~\cite[Thm.~9.1]{Schneider:10c}.

\begin{proposition}\label{Prop:MergeProds}
Let $\dfield{\FF(t_1)\dots(t_e)}{\sigma}$ be a \piE-extension of $\dfield{\FF}{\sigma}$  with $h_i:=\frac{\sigma(t_i)}{t_i}\in\FF^*$. Then there is a \piE-extension $\dfield{\FF(t_1)\dots(t_{e-1})(t)}{\sigma}$ of $\dfield{\FF(t_1)\dots(t_{e-1})}{\sigma}$ with $\frac{\sigma(t)}{t}=h_1^{z_1}\dots h_e^{z_e}$ with $z_1\dots,z_e\in\ZZ$ where $z_e\neq0$. 
\end{proposition}
 

\begin{example}\label{Exp:Sum1Part1}
We construct a difference ring to represent the products arising on the left hand side of~\eqref{Equ:FirstSum}. We start with the \pisiSE-field $\dfield{\KK(k)}{\sigma}$ over $\KK=\QQ(n)$ with $\sigma(k)=k+1$.\\
(1) Take $\binom{n}{k}=\prod_{i=1}^k\frac{n-i+1}{i}$. Since $n-k-1,k\in\ZZ[n,k]$ are irreducible and shift-prime, we can construct, e.g., the \piE-extension $\dfield{\QQ(n)(k)(t_1)(t_2)}{\sigma}$ of $\dfield{\QQ(n)(k)}{\sigma}$ with $\sigma(t_1)=(k+1)\,t_1$ and $\sigma(t_2)=(n-k)t_2$. In this way $\binom{n}{k}$ can be rephrased by $t_2/t_1$. Note that this construction can be simplified with Proposition~\ref{Prop:MergeProds}. Namely, by merging $b:=t_2/t_1$ we get the \piE-extension $\dfield{\QQ(n)(k)(b)}{\sigma}$ of $\dfield{\QQ(n)(k)}{\sigma}$ with $\sigma(b)=\frac{n-k}{k+1}\,b$ and we represent $\binom{n}{k}$ by $b$.\\
(2) We represent $(-1)^k$ with $x$ in the \rE-extension $\dfield{\QQ(n)(k)(b)[x]}{\sigma}$ of $\dfield{\QQ(n)(k)(b)}{\sigma}$ with $\sigma(x)=-x$.
\end{example}

\section{\sigmaSE-extensions and the representation of indefinite nested sums (over products)}\label{Sec:SigmaE}

So far we introduced \sigmaSE-extensions in the setting of difference fields; see Definition~\ref{Def:PiSi}. In the following we introduce a slightly less general definition in the setting of difference rings that excludes the exotic case that sums might pop up in denominators. 
On the other side, we gain more flexibility, since in the ring setting we can handle in addition indefinite nested sums in which objects like $(-1)^k$ occur.
Similar to the field setting, we utilize the following lemma. 

\begin{lemma}\label{Lemma:UniqueSumExt}
Let $\dfield{\AA}{\sigma}$ be a difference ring with $\beta\in\AA$. Let $t$ be transcendental over $\AA$, i.e., $\AA[t]$ is a polynomial ring. Then there is a unique difference ring extension $\dfield{\AA[t]}{\sigma}$ of $\dfield{\AA}{\sigma}$ with $\sigma(t)=t+\beta$.
\end{lemma}

A (polynomial) \sigmaSE-extension is then such a difference ring extension where the constants remain unchanged.

\begin{definition}
A difference ring extension $\dfield{\AA[t]}{\sigma}$ of $\dfield{\AA}{\sigma}$ is called a (polynomial) \sigmaSE-extension if $t$ is transcendental over $\AA$, $\sigma(t)-t\in\AA$ and $\const{\AA[t]}{\sigma}=\const{\AA}{\sigma}$. A tower of such \sigmaSE-extensions is called (nested) \sigmaSE-extension.
\end{definition}

Theorem~\ref{Thm:SigmaExt} stated below is the key tool to represent indefinite nested sums algorithmically within a tower of \sigmaSE-extensions. For the corresponding field version we refer to~\cite{Karr:81}. We remark that Theorem~\ref{Thm:SigmaExt} has been elaborated already in~\cite{Schneider:14b} in a more general setting. But, in this specialized form the following simplified proof is possible.

\noindent By convention, we set $\deg(0):=-\infty$. We start with

\begin{lemma}\label{LemmaB}
Let $\dfield{\AA[t]}{\sigma}$ be a difference ring extension of $\dfield{\AA}{\sigma}$ with $t$ being transcendental over $\AA$, $\sigma(t)=t+\beta$ for some $\beta\in\AA$, and $\KK=\const{\AA}{\sigma}$ being a field.  If there is a
$g\in\AA[t]$ with $\deg(g)\geq1$ and
$\deg(\sigma(g)-g)<\deg(g)-1$,
then there is a $\gamma\in\AA$ with
$\sigma(\gamma)-\gamma=\beta$.
\end{lemma}
\begin{proof}
Let $g=\sum_{i=0}^ng_it^i\in\AA[t]$ with $\deg(g)=n\geq1$ as stated in the lemma, and define $f:=\sigma(g)-g\in\AA[t]$. By $\deg(\sigma(g)-g)<\deg(g)-1$ we get $f=\sum_{i=0}^{n-2}f_i\,t^i$.
Comparing the $n$th and $(n-1)$th coefficient in
$\sum_{i=0}^{n-2}f_i\,t^i=f=\sigma(g)-g=\sum_{i=0}^n\sigma(g_i)(t+\beta)^i-\sum_{i=0}^ng_it^i$
and using $(t+\beta)^i=\sum_{j=0}^i\binom{i}{j}t^{i-j}\beta^j$ for
$0\leq i\leq n$ yield
\begin{align*}
\sigma(g_n)-g_n&=0,&
\sigma(g_{n-1})+\sigma(g_n)\tbinom{n}{1}\beta-g_{n-1}&=0.
\end{align*}
Hence $g_n\in\KK^*$. 
Multiplying the second equation with $\frac1{-n\,g_n}\in\KK^*$ yields
$\sigma(\gamma)-\gamma=\beta$ with $\gamma:=\frac{-g_{n-1}}{n\,g_n}\in\AA$.
\end{proof}

\noindent Now we are in the position to get the following

\begin{theorem}\label{Thm:SigmaExt}
Let $\dfield{\AA[t]}{\sigma}$ be a difference ring extension of $\dfield{\AA}{\sigma}$ with $t$ being transcendental over $\AA$, $\sigma(t)=t+\beta$ for some $\beta\in\AA$, and $\const{\AA}{\sigma}$ being a field. Then this is a \sigmaSE-extension iff there is no $g\in\AA$ with $\sigma(g)=g+\beta$.
\end{theorem}
\begin{proof}
Suppose there is a $g\in\AA$ with $\sigma(g)=g+\beta$. Then $\sigma(t-g)=t-g$, i.e., $t-g\in\const{\AA[t]}{\sigma}\setminus\const{\AA}{\sigma}$. Hence $\dfield{\AA[t]}{\sigma}$ is not a \sigmaSE-extension of $\dfield{\AA}{\sigma}$. Conversely, suppose that there is a $g\in\AA[t]\setminus\AA$ with $\sigma(g)=g$, i.e., $\sigma(g)-g=f$ with $f=0$. Then $\deg(g)>0>-\infty=\deg(f)$. Hence we apply Lemma~\ref{LemmaB} and conclude that there is a $\gamma\in\AA$ with $\sigma(\gamma)=\gamma+\beta$.
\end{proof}

\subsection{Representation of indefinite nested sums (over products)}

Theorem~\ref{Thm:SigmaExt} can be used as follows. Suppose that we are given an expression $F(k)$ in terms of indefinite nested sums defined over objects that can be represented in $\dfield{\AA}{\sigma}$. Suppose in addition that we can solve Problem T (or more generally, Problem PT) in $\dfield{\AA}{\sigma}$ or in a tower of \sigmaSE-extensions over $\dfield{\AA}{\sigma}$.
Then we are able to construct a polynomial \sigmaSE-extension in which the expression $F(k)$ can be represented. 

\begin{example}[Cont.~Ex.~\ref{Exp:Sum1Part1}]\label{Exp:Sum1Part2}
Consider the sum on the left hand side of~\eqref{Equ:FirstSum} and let $F(k)$ be its summand. Take the already constructed difference ring $\dfield{\AA}{\sigma}$ with $\AA=\QQ(n)(k)(b)[x]$ where $x$ and $b$ represent $(-1)^k$ and $\binom{n}{k}$, respectively. We proceed as follows to represent $F(k)$ in a difference ring.\\
(1) By solving Problem T in $\dfield{\AA}{\sigma}$ for the summand $b$ we conclude that there is no $g\in\AA$ with $\sigma(g)=g+b$. Hence by Theorem~\ref{Thm:SigmaExt} there is the \sigmaSE-extension $\dfield{\AA[s]}{\sigma}$ of $\dfield{\AA}{\sigma}$ with $\sigma(s)=s+b$ in which $s$ represents $\sum_{i=0}^{k-1}\binom{n}{i}$.\\
(2) Note that
$f=x\,b\,s$ represents $F(k)$. With \SigmaP\ we solve Problem T in $\dfield{\AA[s]}{\sigma}$ and provide the solution 

\vspace*{-0.2cm}

\begin{equation}\label{Equ:Exp1TeleSol}
g=-\frac{s x (1-k+n)}{b (n+2)}
+\frac{(2 k+1) x}{4 (n+2)};
\end{equation}

\vspace*{-0.2cm}

\noindent for details see Example~\ref{Exp:Sum1Part3}.
Hence replacing $x$ with $(-1)^k$ and $b$ with $\binom{n}{k}$ in $g$ gives the solution $G(k)$ of~\eqref{Equ:TeleExpr} which yields~\eqref{Equ:FirstSum}.
In particular, $\sigma(g)+c\in\AA[s]$ with $c=-\frac1{4(n+2)}$ represents the sum $\sum_{l=1}^k F(l)$; this is exactly the sum on the left hand side of~\eqref{Equ:FirstSum} when $k$ takes over the role of $b$.
\end{example}

In Section~\ref{Sec:PTAlg} we will provide the necessary algorithms that enable the user to perform these constructions algorithmically. In particular, these algorithms will be applicable if we specialize $\dfield{\AA}{\sigma}$ as specified  in Theorem~\ref{Thm:RepresentHyp}. As a consequence, we obtain the following important property.

\smallskip

\noindent\fbox{\begin{minipage}{8.3cm}
Our algorithmic machinery enables one to represent any expression in terms of indefinite nested sums over a finite set of hypergeometric expressions as given in~\eqref{Equ:ProdSet} within the following difference ring. It is built by the \pisiSE-field $\dfield{\KK(k)}{\sigma}$ with $\sigma(k)=k+1$, a tower of \piE-extensions plus possibly one \rE-extension (to represent the products in~\eqref{Equ:ProdSet}), and nested polynomial \sigmaSE-extensions (to represent the indefinite nested sums defined over~\eqref{Equ:ProdSet}).
\end{minipage}}
\section{Algorithms for Problem PT}\label{Sec:PTAlg}

We turn to our algorithmic main result.

\begin{theorem}\label{Thm:MainResult}
Let $\dfield{\FF[x]}{\sigma}$ be an \rE-extension of a difference field $\dfield{\FF}{\sigma}$ and let $\dfield{\EE}{\sigma}$ be a nested polynomial  \sigmaSE-extension of $\dfield{\FF[x]}{\sigma}$. Then the following holds.
\begin{enumerate}
 \item One can solve Problem PT for $\dfield{\EE}{\sigma}$ if one can solve Problem FPLDE in $\dfield{\FF}{\sigma}$. 
 \item This is in particular the case, if $\dfield{\FF}{\sigma}$ is a \pisiSE-field over a constant field as specified in Theorem~\ref{Thm:FPLDEinPiSi}. 
\end{enumerate}
 \end{theorem}

In order to derive this result, we treat first the case of \sigmaSE-extensions. More precisely, let $\dfield{\AA[t]}{\sigma}$ be a \sigmaSE-extension of a difference ring $\dfield{\AA}{\sigma}$ and let $\vect{f}=(f_1,\dots,f_d)\in\AA[t]^d$. Then we want to derive a basis of $V(\vect{f},\AA[t])$. First, we bound the degree of the polynomial solutions. More precisely, define 
$\AA[t]_b=\{h\in\AA[t]|\deg(h)\leq b\}$ for $b\in\ZZ$. Then the following lemma enables one to determine a $b$ with
\begin{equation}\label{Equ:VDegBound}
V(\vect{f},\AA[t])=V(\vect{f},\AA[t]_b).
\end{equation}

\begin{lemma}\label{Lemma:Deg}
Let $\dfield{\AA[t]}{\sigma}$ be a \sigmaSE-extension of $\dfield{\AA}{\sigma}$ with constant field $\KK$. Let $g,f\in\AA[t]$ with $\sigma(g)-g=f$. Then $\deg(g)\leq\max(\deg(f),-1)+1$.
\end{lemma}
\begin{proof}
Suppose that there are $f,g\in\AA[t]$ with $\sigma(g)-g=f$. If $f=0$, then $g\in\const{\AA[t]}{\sigma}=\KK$ and the bound holds. Now suppose that $f\neq0$ and $\deg(g)>\max(\deg(f),-1)+1=\deg(f)+1$.
By Lemma~\ref{LemmaB} there is a $\gamma\in\AA$ with $\sigma(\gamma)-\gamma=\sigma(t)-t$. Hence $\dfield{\AA[t]}{\sigma}$ is not a \sigmaSE-extension of $\dfield{\AA}{\sigma}$ by Theorem~\ref{Thm:SigmaExt}, a contradiction.
\end{proof}

Thus by Lemma~\ref{Lemma:Deg} we can choose
\begin{equation}\label{Equ:DegBPara}
b=\max(\deg(f_1),\dots,\deg(f_d),-1)+1
\end{equation}
such that~\eqref{Equ:VDegBound} holds.
Given this degree bound $b$, we illustrate the underlying strategy to compute a basis of $V(\vect{f},\AA[t]_b)$ by solving several Problems PT in $\dfield{\AA}{\sigma}$.

\begin{example}[Cont.~Ex.~\ref{Exp:Sum1Part2}]\label{Exp:Sum1Part3}
Take our \sigmaSE-extension $\dfield{\AA[s]}{\sigma}$ of $\dfield{\AA}{\sigma}$ with $\AA=\KK(k)(b)[x]$ and constant field $\KK=\QQ(n)$. Solving Problem T for $f=\frac{s\,x}{b}$ is equivalent to compute a basis of $V=V((f),\AA[s])$. By Lemma~\ref{Lemma:Deg} it follows that $V((f),\AA[s])=V((f),\AA[s]_2)$. Let $(c_1,g)\in V$. Thus $g=g_2\,s^2+g_1\,s+g_0$ with $g_2,g_1,g_0\in\AA$ and $c_1\in\KK$. In particular,  
\begin{equation}\label{Equ:Deg2Eq}
\sigma(g_2\,s^2+g_1\,s+g_0)-(g_2\,s^2+g_1\,s+g_0)=c_1\,f.
\end{equation}
By coefficient comparison w.r.t.\ $s^2$ we get that $\sigma(g_2)-g_2=0$, i.e., $c:=g_2\in\KK$. Thus moving $\sigma(g_2\,s^2)-g_2\,s^2=c(b^2
+2 b s)$ in~\eqref{Equ:Deg2Eq} to the right hand side gives
\begin{equation}\label{Equ:Deg1Eq}
\sigma(g_1\,s+g_0)-(g_1\,s+g_0)=
c_1\,f-c(b^2
+2 b s).
\end{equation}
By coefficient comparison w.r.t. $s$ we get the constraint
$\sigma(g_1)-g_1=c_1\frac{x}{b}-c\,2 b.$
Therefore we compute a basis of $V(\tilde{\vect{f}},\AA)$ with 
$\tilde{\vect{f}}=(\frac{x}{b},-2b)$. As will be worked out in Example~\ref{Exp:Sum1Part4}, we get the basis
$\{( 1 , 0 , -\frac{x (1-k+n)}{b (n+2)}),(0,0,1)\}$.
Thus we conclude that $c=0$. Now we plug in the general solution
$g_1=-\frac{x (1-k+n)}{b (n+2)}+d$ with $d\in\KK$ into~\eqref{Equ:Deg1Eq} and get
$$\sigma(g_0)-g_0=c_1\,\tfrac{(-k-1) x}{n+1}-d\,b.$$
Finally, we compute the basis $\{(1 , 0 , \frac{(2 k+1) x}{4 (n+1)}),(0,0,1)\}$
of $V((\frac{(-k-1) x}{n+1},-b),\AA)$. This shows that $d=0$ and we obtain the general solution  $g_0=\frac{(2 k+1) x}{4 (n+1)}+e$ with $e\in\KK$. In total, we get the general solution $c_1\in\KK$ and $g+e$ where $g$ is given in~\eqref{Equ:Exp1TeleSol}. Thus a  basis of $V((f),\AA[s])$ is $\{(1,g),
(0,1)\}$.
\end{example}

This reduction strategy can be summarized with

\begin{theorem}\label{Thm:SigmaLift}
Let $\dfield{\AA[t]}{\sigma}$ be a \sigmaSE-extension of $\dfield{\AA}{\sigma}$ with constant field $\AA$. If one can solve problem PT in $\dfield{\AA}{\sigma}$ then one can solve problem PT in $\dfield{\AA[t]}{\sigma}$.  
\end{theorem}
\begin{proof}
Let $\vect{f}=(f_1,\dots,f_d)\in\AA[t]^d$. Then take $b$ as given in~\eqref{Equ:VDegBound} and it follows that~\eqref{Equ:DegBPara}.
Hence we have to search for all $g=g_0+\dots+g_b\,t^b\in\AA[t]_{b}$ and $c_1,\dots,c_d\in\KK$ with~\eqref{Equ:PT}.
By induction on $b$ (equivalently by recursion) we can now calculate a basis of $V(\vect{f},\AA[t]_b)$.  Note that the highest possible degree in~\eqref{Equ:PT} is bounded by $b$. Thus if $b=0$, a basis can be calculated by assumption (base case).
Otherwise, by coefficient comparison w.r.t.\ $t^b$ (the highest possible power) we get the constraint
$\sigma(g_b)-g_b=c_1\,\tilde{f}_1+\dots+c_d\,\tilde{f}_d$
where $\tilde{f}_i=\coeff(f_i,b)$. By assumption we can compute a basis $\tilde{B}$ of $V(\vect{\tilde{f}},\AA)$ with $\vect{\tilde{f}}=(\tilde{f}_1,\dots,\tilde{f}_d)\in\AA^d$. Plugging in the solution into~\eqref{Equ:PT} yields Problem PT in $\dfield{\AA[t]}{\sigma}$ in terms of the remaining coefficients $g_0,\dots,g_{b-1}$. Note that the highest possible degree in this parameterized telescoping equation is at most $b-1$. Now compute a basis $B'$ for this solution space by the induction assumption (by recursion). Finally, combining the bases $\tilde{B}$ and $B'$ accordingly produces a basis of $V(\vect{f},\AA[t]_b)$. The technical details of this combining step are equivalent to the ones given in~\cite[Cor.~1]{Schneider:14}. 
\end{proof}

Next, we turn to \rE-extensions.

\begin{theorem}
 Let $\dfield{\FF[x]}{\sigma}$ be an \rE-extension of a difference field $\dfield{\FF}{\sigma}$. Then one can solve Problem PT for $\dfield{\FF[x]}{\sigma}$ if one can solve Problem FPLDE in $\dfield{\FF}{\sigma}$. 
\end{theorem}
\begin{proof}
We have $V(\vect{f},\FF[x])=V(\vect{f},\FF[x]_{\lambda-1})$ where $\lambda$ is the order of $x$.
Since the $x^0,\dots,x^{\lambda-1}$ are linearly independent over $\FF$, the same induction argument (recursion) as in the proof of Theorem~\ref{Thm:SigmaLift} is applicable.
\end{proof}

\begin{example}[Cont.~Ex.~\ref{Exp:Sum1Part3}]\label{Exp:Sum1Part4}
Take our \rE-extension $\dfield{\FF[x]}{\sigma}$ of $\dfield{\FF}{\sigma}$ with $\FF=\KK(k)(b)$. We aim at computing a basis of $V=V(\vect{f},\FF[x])$ with 
$\vect{f}=(\frac{x}{b},-2b)$. Let $(c_1,c_2,g)\in V$. Then $g=g_1\,x+g_0\in\FF[x]$ and $c_1,c_2\in\KK$ with
\begin{equation}\label{Equ:XDeg1Constr}
\sigma(g_1\,x+g_0)-(g_1\,x+g_0)=c_1\,\frac{x}{b}-c\,2b.
\end{equation}
By coefficient comparison w.r.t.\ $x$ we conclude that $\sigma(g_1)+g_1=c_1\,\frac{1}{b}+c_2\,0$. Therefore we compute a basis of $V((1,1),(\frac{1}{b},0),\FF)$ by solving Problem FPLDE (see Theorem~\ref{Thm:FPLDEinPiSi}; here one could also use a variation of the algorithm given in~\cite{Zeilberger:91}): we get the basis
$\{(1,0,-\frac{(1-k+n)}{b(n+2)}),(0,1,0)\}$. Hence the general solution is $g_1=-\frac{(1-k+n)}{b(n+2)}$. Plugging the found $g_1$ into~\eqref{Equ:XDeg1Constr} gives the constraint
$\sigma(g_0)-g_0=c_1\,0-c_2\,2\,b$.
We calculate the basis $\{(1,0,0),(0,0,1)\}$ of $V((0,2b),\AA)$, i.e., we get the general solution $g_0=d$ with $d\in\KK$, $c_1\in\KK$ and $c_2=0$. Combining these solutions produces $g=-\frac{x (1-k+n)}{b (n+2)}+d$ and $c_1\in\KK$. Thus a basis of $V(\vect{f},\FF[x])$ is 
$\{( 1 , 0 , -\frac{x (1-k+n)}{b (n+2)}),(0,0,1)\}$.
\end{example}

Now suppose that we are given a nested polynomial \sigmaSE-extension $\dfield{\EE}{\sigma}$ of a difference ring $\dfield{\AA}{\sigma}$. Hence by iterative application of Theorem~\ref{Thm:SigmaLift} one can solve Problem PT in $\dfield{\EE}{\sigma}$ if one can solve Problem PT in $\dfield{\AA}{\sigma}$. In particular,
setting $\AA=\FF[x]$ together with Theorem~\ref{Thm:FPLDEinPiSi} proves Theorem~\ref{Thm:MainResult}. 

We conclude this article with a more challenging summation problem.
Denote the sum on the left hand side of~\eqref{Equ:SecondSum} by $S(n)$ and its summand by $F(n,k)$.\\ 
(1) By Example~\ref{Exp:Sum1Part1} we get $\dfield{\AA}{\sigma}$ with $\AA=\QQ(n)(k)(b)[x]$ in which we can represent $\binom{n}{k}$ with $b$ and $(-1)^k$ with $x$.\\
(2) Next, we check that there is no $g\in\AA$ with $\sigma(g)-g=\frac{-x}{k+1}$. Thus we cannot express $a(k)=\sum_{i=1}^k\frac{(-1)^i}{i}$ in $\AA$. However, we can adjoin it in form of the \sigmaSE-extension $\dfield{\AA[s]}{\sigma}$ of $\dfield{\AA}{\sigma}$ with $\sigma(s)=s+\frac{-x}{k+1}$.\\
(3) Similarly, we fail to find a $g\in\AA[s]$ with $\sigma(g)=g+b\,s$. Hence we could adjoin $\sum_{l=0}^k F(n,l)$ in form of a \sigmaSE-extension; but this amounts to no simplification.\\ 
Hence we proceed differently by using Zeilberger's creative telescoping paradigm: we set $F_i(k)=F(n+i-1,k)$ for $1\leq i\leq d$ and
search for a solution of~\eqref{Equ:PClosedForm} with $d=1,...$. We skip the case $d=1$, which is equivalent to telescoping, and continue with $d=2$. Note that 
$F_1(k)=\binom{n}{k}a(k)$ and $F_2(k)=F(n+1,k)=\frac{n+1}{n-k+1}\binom{n}{k}a(k)$ can be represented by $f_1=b\,s$ and $f_2=\frac{n+1}{n-k+1}b\,s$, respectively. Hence hunting for all creative telescoping solutions in $\dfield{\AA[s]}{\sigma}$ is equivalent to computing a basis of $V((f_1,f_2),\AA[s])$. In this particular instance \SigmaP\ computes the non-trivial basis $\{(-2 , 1 , \frac{b x}{n+1}
        +\frac{b k s}{-1
        +k
        -n
        }),(0,0,1)\}$.
Therefore we obtain the solution $c_1=-2$, $c_2=1$ and $G(k)= \frac{1}{n+1}(-1)^k\binom{n}{k}
        +\frac{k}{-1
        +k
        -n
        }\binom{n}{k}a(k)$
of~\eqref{Equ:PClosedForm}.
Summing~\eqref{Equ:PClosedForm} over $k$ from $0$ to $n$ and taking care of compensating terms leads to the recurrence
$$S(n+1)-2\,S(n)=-\frac{1}{n+1}.$$
In this instance, one can read off the right hand side of~\eqref{Equ:SecondSum} by the variation of constants method. For recurrence relations of higher-order more advanced algorithms~\cite{Abramov:94,Singer:99,Schneider:05a,ABPS:14} are available  within the summation package~\SigmaP.

\section*{Acknowledgement}
I would like to thank the organizers of SYNASC 2014 for their hospitality and this wonderful event in Timi\c soara. This work was supported by the Austrian Science Fund (FWF) grant SFB F50 (F5009-N15) and the European
Commission through contract PITN-GA-2010-264564 ({LHCPhenoNet}).







\end{document}